\documentclass{article}
\usepackage[hidelinks]{hyperref}
\usepackage{amsmath}
\usepackage{parskip} 
\usepackage{natbib}
\usepackage{mathrsfs}
\usepackage{pdfpages}
\usepackage{bm}
\usepackage{booktabs}
\usepackage{graphics}
\usepackage{multicol}
\usepackage[bf]{caption}
\usepackage{color}
\usepackage{amssymb}
\usepackage{multirow}
\usepackage{booktabs}
\usepackage{subcaption}

\usepackage{comment}
\usepackage{color}

\usepackage{amsthm}
\newtheorem{proposition}{Proposition}[section]
\usepackage[margin=1in]{geometry}

\title{On conditions for long-wave equivalent medium to be isotropic and on analysis of parameters indicating anisotropy of equivalent TI medium}
\author{Filip P. Adamus\footnote{
Department of Earth Sciences, Memorial University of Newfoundland, Canada, {\tt adamusfp@gmail.com}}}

\date{}
\begin{document}
\maketitle
\begin{abstract}
In this paper, we consider a long-wave equivalent medium to a finely parallel-layered inhomogeneous medium, obtained using the \citet{Backus} average.
Following the work of~\citet{Postma} and~\citet{Backus}, we show explicitly the derivations of the conditions to obtain the equivalent isotropic medium. 
We demonstrate that there cannot exist a transversely isotropic (TI) equivalent medium with the coefficients $c^{\overline{\rm TI}}_{1212} \neq c^{\overline{\rm TI}}_{2323}$, $c^{\overline{\rm TI}}_{1111} = c^{\overline{\rm TI}}_{3333}$ and $c^{\overline{\rm TI}}_{1122} = c^{\overline{\rm TI}}_{1133}$. 
Moreover, we consider a new parameter, $\varphi$, indicating the anisotropy of the equivalent medium, and we show its range and properties.
Subsequently, we compare $\varphi$ to the~\citet{Thomsen} parameters, emphasizing its usefulness as a supportive parameter showing the anisotropy of the equivalent medium or as an alternative parameter to $\delta$.
We argue with certain~\citet{Berrymanetal} considerations regarding the properties of the anisotropy parameters $\epsilon$ and $\delta$.
Additionally, we show an alternative way---to the one mentioned by~\citet{Berrymanetal}---of indicating changing fluid content in layered Earth.
\end{abstract}
\section{Introduction}
The problem of fine, parallel layering and its long-wave equivalent medium approximation has been treated by a number of authors, among many of them are: \citet{Postma}, \citet{Backus}, \citet{HelbigandSchoenberg}, \citet{SchoenbergandMuir}, \citet{Berrymanetal}, and \citet{Bosetal}. 
In the work of~\citeauthor{Postma}, the periodic, isotropic, two-layered medium (PITL) was considered; the formula for the equivalent homogeneous transversely isotropic (TI) medium was shown.
Also, the isotropic condition for such a medium was found, namely, Lam\'e parameter $\mu=\rm{const}$. 
In this paper, we persue~\citeauthor{Postma}'s consideration, adding some new insights and expressing it in a modern notation.

In~\citeauthor{Backus}'s work, a finely-layered, horizontally-stratified, TI medium with no periodicity assumption was considered.
An elegant formula of the TI medium, long-wave equivalent to isotropic or TI layers of different thicknesses, was shown.
Also, similarly to \citeauthor{Postma}, the author has implicitly derived that the constant $\mu$ is responsible for the isotropy of the equivalent medium, however with no assumption of periodicity.
We show explicitly the aforementioned derivation, using modern notation.

Finally, inspired by the work of \citeauthor{Berrymanetal}, we revise the consideration about the \citeauthor{Thomsen} parameter $\epsilon$ and its dependance on $\mu$. 
Also, we analyze the behaviour of $\delta$ in the context of variations of Lam\'e parameter $\lambda$.
Moreover, we examine an alternative way---to the one shown by \citeauthor{Berrymanetal}---of indicating the change of fluid content in layered Earth.
To do so, we consider another parameter indicating the anisotropy of the equivalent medium, find its range and compare its properties to the all three \citeauthor{Thomsen} parameters.
For simplicity, throughout this paper, we denote long-wave equivalent TI or isotropic medium---received using the \citeauthor{Backus} average---as the~\citeauthor{Backus} medium or the equivalent medium.
\section{Background}
\subsection{Backus average}
As shown by~\citeauthor{Backus}, a medium composed of parallel isotropic layers, whose individual thicknesses are much shorter than the wavelength, respond---to the wave propagation---as a single, homogeneous and transversely isotropic medium. 
The elasticity parameters of this medium are 
\begin{equation}   \label{bac1}
    c_{1111}^{\overline{\rm TI}}=\overline{\left(\frac{\lambda}{\lambda+2\mu}\right)}^2\,\overline{\left(\frac{1}{\lambda+2\mu}\right)}^{-1}+\overline{\left(\frac{4(\lambda+\mu)\mu}{\lambda+2\mu}\right)}\,,
\end{equation}
\begin{equation}   \label{bac2}
    c_{1122}^{\overline{\rm TI}}=\overline{\left(\frac{\lambda}{\lambda+2\mu}\right)}^2\,\overline{\left(\frac{1}{\lambda+2\mu}\right)}^{-1}+\overline{\left(\frac{2\lambda\mu}{\lambda+2\mu}\right)}\,,
\end{equation}
\begin{equation}  \label{bac3}
    c_{1133}^{\overline{\rm TI}}=\overline{\left(\frac{\lambda}{\lambda+2\mu}\right)}\,\overline{\left(\frac{1}{\lambda+2\mu}\right)}^{-1}\,,
\end{equation}
\begin{equation}  \label{bac4}
     c_{3333}^{\overline{\rm TI}}=\overline{\left(\frac{1}{\lambda+2\mu}\right)}^{-1}\,,
 \end{equation}
 \begin{equation}  \label{bac5}
     c_{2323}^{\overline{\rm TI}}=\overline{\left(\frac{1}{\mu}\right)}^{-1}\,,
 \end{equation}
  \begin{equation}  \label{bac6}
   c_{1212}^{\overline{\rm TI}}=\overline{\mu}\,,
 \end{equation}
where $\lambda$ and $\mu$ are the Lam\'e parameters for each layer and the bar denotes the weighted average.
The weight is the layer thickness; herein, since all layers have the same thickness, we have the arithmetic average.
A TI medium, whose rotation symmetry axis is parallel to the $x_3$-axis, is~\citep[see e.g.,][p.134]{SlawinskiGreen}
\begin{equation*}
C^{\overline{\rm TI}}=
\left[
\begin{array}{cccccc}
c^{\overline{\rm TI}}_{1111} & c^{\overline{\rm TI}}_{1122} & c^{\overline{\rm TI}}_{1133} & 0 & 0 & 0\\
c^{\overline{\rm TI}}_{1122} & c^{\overline{\rm TI}}_{1111} & c^{\overline{\rm TI}}_{1133} & 0 & 0 & 0\\
c^{\overline{\rm TI}}_{1133} & c^{\overline{\rm TI}}_{1133} & c^{\overline{\rm TI}}_{3333} & 0 & 0 & 0\\
0 & 0 & 0 & 2c^{\overline{\rm TI}}_{2323} & 0 & 0\\
0 & 0 & 0 & 0 & 2c^{\overline{\rm TI}}_{2323} & 0\\
0 & 0 & 0 & 0 & 0 & 2c^{\overline{\rm TI}}_{1212}
\end{array}
\right]
\,,
\end{equation*}
where $c_{1122}^{\overline{\rm TI}}=c_{1111}^{\overline{\rm TI}}-2c_{1212}^{\overline{\rm TI}}$, which means that expressions~(\ref{bac1})--(\ref{bac6}) consist of five independent parameters.

\subsection{Thomsen parameters}
To examine anisotropy strength of transversely isotropic tensor, we use~\citeauthor{Thomsen} parameters,
 \begin{equation*}\label{eq:thom1}
     \gamma := \frac{c_{1212}^{\overline{\rm TI}}-c_{2323}^{\overline{\rm TI}}}{2c_{2323}^{\overline{\rm TI}}}\,,
 \end{equation*}
\begin{equation*}\label{eq:thom2}
     \delta := \frac{\left(c_{1133}^{\overline{\rm TI}}+c_{2323}^{\overline{\rm TI}}\right)^2-\left(c_{3333}^{\overline{\rm TI}}-c_{2323}^{\overline{\rm TI}}\right)^2}{2c_{3333}^{\overline{\rm TI}}\left(c_{3333}^{\overline{\rm TI}}-c_{2323}^{\overline{\rm TI}}\right)}\,,
 \end{equation*}
\begin{equation*}\label{eq:thom3}
     \epsilon := \frac{c_{1111}^{\overline{\rm TI}}-c_{3333}^{\overline{\rm TI}}}{2c_{3333}^{\overline{\rm TI}}}\,.
 \end{equation*}
 As shown by~\citet{Adamusetal}---for the \citeauthor{Backus} average---growing anisotropy of an equivalent medium implies the growth of inhomogeneity among layers. 
 \subsection{Stability conditions}
The allowable relations among the elasticity parameters are stated by the stability conditions that express the fact that it is necessary to expend energy to deform a material~\citep[e.g.,][Section 4.3]{SlawinskiRed}. 
These conditions mean that every elasticity tensor must be positive-definite, wherein a tensor is positive-definite if and only if all its eigenvalues are positive.
For any isotropic elasticity tensor, the inequalities
\begin{equation*}
\label{ineq}
    c_{1111}\geq\tfrac{4}{3}\,c_{2323}\geq0
\end{equation*}
or, in a different notation, using Lam\'e parameters,
\begin{equation}
\label{eq:stab0}
    \lambda\geq-\tfrac{2}{3}\,\mu  
\qquad    
\mathrm{and}
\qquad
\mu\geq0
\end{equation}
ensure that all eigenvalues are positive, thus the stability conditions are satisfied.
Any transversely isotropic elasticity tensor, to satisfy the stability conditions, must obey the inequalities
\begin{equation}\label{eq:stab}
c_{1212}^{\overline{\rm TI}}\geq0\,,
\qquad
c_{2323}^{\overline{\rm TI}}\geq0\,,
\qquad
c_{3333}^{\overline{\rm TI}}\geq0\,,
\qquad
c_{1122}^{\overline{\rm TI}}+c_{1212}^{\overline{\rm TI}}\geq0\,,
\end{equation}
\begin{equation*}
\left(c_{1122}^{\overline{\rm TI}}+c_{1212}^{\overline{\rm TI}}\right)c_{3333}^{\overline{\rm TI}}\geq\left(c_{1133}^{\overline{\rm TI}}\right)^2\,.
\end{equation*}
\section{Conditions for Backus medium to be isotropic} 
An equivalent medium to be isotropic, must satisfy
\begin{equation}\label{eq:cond1}
c^{\overline{\rm TI}}_{1212} = c^{\overline{\rm TI}}_{2323}\,,
\end{equation}
\begin{equation}\label{eq:cond2}
c^{\overline{\rm TI}}_{1111} = c^{\overline{\rm TI}}_{3333}\,
\end{equation}
and 
\begin{equation}\label{eq:cond3}
c^{\overline{\rm TI}}_{1133} = c^{\overline{\rm TI}}_{1111} - 2\,c^{\overline{\rm TI}}_{1212}\,. 
\end{equation}
The last relation comes from $c^{\overline{\rm TI}}_{1122} = c^{\overline{\rm TI}}_{1133}$ and from the constraint  $c^{\overline{\rm TI}}_{1111} = c^{\overline{\rm TI}}_{1122} + 2\,c^{\overline{\rm TI}}_{1212}$\,.
\subsection{Conditions for isotropy: PITL medium} 
\label{sec:pitl}
Condition~(\ref{eq:cond1}) is satisfied if the stiffness factor, $\mu$, is constant. 
To discuss the other conditions, let us consider a PITL medium.
Such a medium is a periodic structure consisting of alternating plane, parallel, isotropic and homogeneous elastic layers, with constant thickness and density.
A PITL medium has only four Lam\'e parameters, namely, $\lambda_1$, $\lambda_2$, $\mu_1$ and $\mu_2$, which simplifies the averaging process.
After laborious computations, we see that condition~(\ref{eq:cond2}) is obeyed if 
\begin{equation*}
    \left(\mu_2-\mu_1\right)\left(\mu_2-\mu_1+\lambda_2-\lambda_1\right)=0\,
\end{equation*} and
it is true for $\mu_1=\mu_2$ or $\lambda_1+\mu_1=\lambda_2+\mu_2$.

Condition~(\ref{eq:cond3}) is satisfied if
\begin{equation}\label{eq:b=f}
    \lambda_1\mu_1+\lambda_2\mu_2=\lambda_1\mu_2+\lambda_2\mu_1\,,
\end{equation}
hence it is true for $\mu_1=\mu_2$ or $\lambda_1=\lambda_2$.

We conclude that for a PITL medium, $\mu$ being constant is a unique and necessary condition for equivalent medium to be isotropic (in a different manner it was also shown by \citeauthor{Postma}).
If condition~(\ref{eq:cond2}) is satisfied by $\lambda_1+\mu_1=\lambda_2+\mu_2$, then in order to simultaneously satisfy condition~(\ref{eq:cond3}), we transform equation~(\ref{eq:b=f}) to 
\begin{equation*}
    \left(\lambda_2+\mu_2-\mu_1\right)\mu_1+\lambda_2\mu_2=\left(\lambda_2+\mu_2-\mu_1\right)\mu_2+\lambda_2\mu_1
\end{equation*}
and, after simple computation, we receive
\begin{equation*}
    \left(\mu_1-\mu_2\right)^2=0\,,
\end{equation*}
which means that $\mu$ must be constant.
Hence, if both conditions~(\ref{eq:cond2}) and (\ref{eq:cond3}) are obeyed, then it automatically follows that condition~(\ref{eq:cond1}) is also satisfied and a medium is isotropic.
In other words, there is no TI equivalent medium, in which
\begin{equation*} 
c^{\overline{\rm TI}}_{1212} \neq c^{\overline{\rm TI}}_{2323}\,,\qquad c^{\overline{\rm TI}}_{1111} = c^{\overline{\rm TI}}_{3333}\qquad \mathrm{and} \qquad c^{\overline{\rm TI}}_{1133} = c^{\overline{\rm TI}}_{1111} - 2\,c^{\overline{\rm TI}}_{1212}\,.
\end{equation*}
Also, it can be easily deduced that a TI equivalent medium, in which 
\begin{equation*}
c^{\overline{\rm TI}}_{1212} \neq c^{\overline{\rm TI}}_{2323}\,,\qquad c^{\overline{\rm TI}}_{1111} = c^{\overline{\rm TI}}_{3333}\qquad \mathrm{and} \qquad c^{\overline{\rm TI}}_{1133} \neq c^{\overline{\rm TI}}_{1111} - 2\,c^{\overline{\rm TI}}_{1212}
\end{equation*}
may exist.
Moreover, if $\lambda_1=\lambda_2$, we deduce that there may exist a TI \citeauthor{Backus} medium with 
\begin{equation*}
c^{\overline{\rm TI}}_{1212} \neq c^{\overline{\rm TI}}_{2323}\,,\qquad c^{\overline{\rm TI}}_{1111} \neq c^{\overline{\rm TI}}_{3333}\qquad \mathrm{and} \qquad c^{\overline{\rm TI}}_{1133} = c^{\overline{\rm TI}}_{1111} - 2\,c^{\overline{\rm TI}}_{1212}\,.
\end{equation*}
\subsection{Conditions for isotropy: finely layered medium}
\label{sec:CondIsoFLM}
Let us consider the general case of thin parallel layers without any assumption of periodicity.
Invoking the work of \citeauthor{Backus}, we can rewrite equations~(\ref{bac1})--(\ref{bac6}) as
\begin{equation}\label{eq:ABCF1}
c_{1111}^{\overline {\rm TI}}=:A=B+2M\,,
\end{equation}
\begin{equation}
c_{1122}^{\overline {\rm TI}}=:B=2M-4S+R^{-1}\left(1-2T\right)^2\,,
\end{equation}
\begin{equation}
c_{1133}^{\overline {\rm TI}}=:F=R^{-1}\left(1-2T\right)\,,
\end{equation}
\begin{equation}
c_{3333}^{\overline {\rm TI}}=:C=R^{-1}\,,
\end{equation}
\begin{equation}
c_{2323}^{\overline {\rm TI}}=:L\,,
\end{equation}
\begin{equation}\label{eq:ABCF2}
c_{1212}^{\overline {\rm TI}}=:M\,,
\end{equation}
where
\begin{equation*}
R=\overline{\left(\frac{\mu}{\lambda+2\mu}\right)\left(\frac{1}{\mu}\right)}\,,\qquad
S=\overline{\left(\frac{\mu}{\lambda+2\mu}\right)\mu}\,,\qquad
T=\overline{\left(\frac{\mu}{\lambda+2\mu}\right)}\,.
\end{equation*}
Subsequently, \citeauthor{Backus} has expressed the conditions for the equivalent medium to be isotropic, namely, 
\begin{equation}\label{eq:baccond}
M=L\,,\qquad S=MT \qquad \mathrm{and} \qquad T=MR\,,
\end{equation}
showing that the aforementioned conditions are satisfied if and only if $\mu$ is constant.
Let us now prove that the conditions for isotropy~(\ref{eq:baccond}) satisfy equations~(\ref{eq:cond1})--(\ref{eq:cond3}), hence that \citeauthor{Backus}'s consideration is correct.
\newline
\begin{proposition}\label{prop}
Conditions $M=L$, $S=MT$ and $T=MR$ are satisfied if and only if $c^{\overline{\rm TI}}_{1212} = c^{\overline{\rm TI}}_{2323}$\,, $c^{\overline{\rm TI}}_{1111} = c^{\overline{\rm TI}}_{3333}$\, and \,$c^{\overline{\rm TI}}_{1133} = c^{\overline{\rm TI}}_{1111} - 2\,c^{\overline{\rm TI}}_{1212}$\,.
\end{proposition}
\begin{proof}

$M=L$ is explicitly equation~(\ref{eq:cond1}).
Using equations~(\ref{eq:ABCF1})--(\ref{eq:ABCF2}), after some algebraic rearrangements, both $S=MT$ and $T=MR$ can be rewritten as
\begin{equation*}
\begin{cases}
(4C)^{-1}(F^2+2MC-BC)=M(2C)^{-1}(C-F)\\
(2C)^{-1}(C-F)=MC^{-1}
\end{cases}\,,
\end{equation*}
then we have
\begin{equation}\label{eq:proof}
\begin{cases}
F^2+2MF-BC=0\\
2M=C-F
\end{cases}
\end{equation}
and after substitution, we receive
\begin{equation*}
F^2+(C-F)F-BC=0\,,
\end{equation*}
which is
\begin{equation*}
C(F-B)=0\,.
\end{equation*}
Considering strict stability conditions, $C>0$, as opposed to weak ones from equation~(\ref{eq:stab}), we see that $B=F$, which is equation~(\ref{eq:cond3}).
After inserting $B=F$ into equation~(\ref{eq:proof}), we receive $C=B+2M$. 
Knowing that, for TI medium there is a constraint $A=B+2M$, we receive another relation $A=C$, which is equation~(\ref{eq:cond2}).
\end{proof}
Hence, assuming strict stability conditions, conditions~(\ref{eq:baccond}) are the conditions for isotropy, and it is proved that $\mu=\rm{const}$ is a necessary condition for any finely-layered medium to be isotropic.
From the considerations in Section~\ref{sec:pitl} and Proposition~\ref{prop}, we see that to verify if the \citeauthor{Backus} medium is isotropic, we check if both $S=MT$ and $T=MR$, or only if $M=L$.
\section{Anisotropy parameters}
\subsection{Anisotropy parameter $\varphi$}
\label{sec:VarPhi}
Let us define anisotropy parameter
\begin{equation}\label{phi}
\varphi:=\frac{c_{1122}^{\overline{\rm TI}}-c_{1133}^{\overline{\rm TI}}}{2c_{1122}^{\overline{\rm TI}}}\,\,.
\end{equation}
This parameter is equal to zero in the case of $c_{1122}^{\overline {\rm TI}}=c_{1133}^{\overline {\rm TI}}$, which occurs---as it was shown in Section~\ref{sec:pitl}---if $ \lambda_1\mu_1+\lambda_2\mu_2=\lambda_1\mu_2+\lambda_2\mu_1$. 
Therefore, $\varphi=0$ for constant rigidity (then the Backus medium is isotropic) or for constant $\lambda$.
The advantage of introducing such a parameter is that it gives us an insight into the variability of $\lambda$.
Also, similarly to \citeauthor{Thomsen} parameters, it has zero value in the case of isotropy.
It may be treated as a replacement of \citeauthor{Thomsen} parameter $\delta$\,.
Similarly to $\delta$---and in contrast to $\epsilon$ and $\gamma$---it depends on $c_{1133}^{\overline{\rm TI}}$\,.
Furthermore, it facilitates an examination of certain properties of the medium, which are not emphasized by~$\delta$\,, even though, mathematically, adding another parameter is unnecessary, since they contain the information about the five independent TI elasticity parameters.
Physically, however, it can add additional subtle information---as it is discussed in Section~\ref{sec:fluids}---useful in the detection of fluid change in layered Earth.

We use form~\eqref{phi} to have a larger possibility of receiving a positive value instead of a negative one, since typically, $c_{1122}^{\overline{\rm TI}} > c_{1133}^{\overline{\rm TI}}$ and $c_{1122}^{\overline{\rm TI}} > 0$.
Also, for fluid detection, it is more convenient to use expression~(\ref{phi}) than
\begin{equation*}
\frac{c_{1133}^{\overline{\rm TI}}-c_{1122}^{\overline{\rm TI}}}{2c_{1122}^{\overline{\rm TI}}}\,.
\end{equation*}
The choice of the denominator is motivated by the fact that we want $\varphi$ to be unitless and to give the relative value.
Similarly to \citeauthor{Thomsen} parameters, we multiply the denominator by two.

Let us check the range of $\varphi$.
For the two layer case, we receive
\begin{equation}\label{eq:phitwolay}
\varphi=\frac{\lambda_1\mu_2+\lambda_2\mu_1-\lambda_1\mu_1-\lambda_2\mu_2}{4\left(\lambda_1\lambda_2+\lambda_1\mu_2+\lambda_2\mu_1\right)}=\frac{\left(\mu_1-\mu_2\right)\left(\lambda_2-\lambda_1\right)}{4\left(\lambda_1\lambda_2+\lambda_1\mu_2+\lambda_2\mu_1\right)}
\end{equation}
and, confirming the result from equation~(\ref{eq:b=f}), for $\lambda_1=\lambda_2$, we have $\varphi=0$.
We try to determine the range of $\varphi$, by checking the boundary values of $\lambda_1$, $\lambda_2$, $\mu_1$ and $\mu_2$.
According to condition~(\ref{eq:stab0}), the smallest value of $\lambda_2$ is $-\frac{2}{3}\mu_2$, hence, we receive
\begin{equation*}
\varphi=\frac{3\left(\mu_1-\mu_2\right)\left(-\lambda_1-\tfrac{2}{3}\mu_2\right)}{4\left(\lambda_1\mu_2-2\mu_2\mu_1\right)}\,,
\end{equation*}
where $\lambda_1$ may vary from $-\tfrac{2}{3}\mu_1$ to $\infty$.
For $\lambda_1=-\tfrac{2}{3}\mu_1$, we obtain
\begin{equation}\label{eq:range1}
\varphi=\frac{\left(\mu_1-\mu_2\right)\left(\tfrac{2}{3}\mu_1-\tfrac{2}{3}\mu_2\right)}{-\tfrac{32}{9}\mu_1\mu_2}=-\frac{3}{16}\frac{\left(\mu_1-\mu_2\right)^2}{\mu_1\mu_2}\,,
\end{equation}
which gives the range $-\infty\leq\varphi\leq -\frac{3}{16}$.
For $\lambda_1$ tending to $\infty$, we obtain the limit
\begin{equation}\label{eq:range2}
\varphi=\lim_{\lambda_1\rightarrow\infty}\frac{\left(\mu_1-\mu_2\right)\left(-\lambda_1-\tfrac{2}{3}\mu_2\right)}{\tfrac{4}{3}\left(\lambda_1\mu_2-2\mu_2\mu_1\right)}=\frac{3}{4}\left(\frac{\mu_2-\mu_1}{\mu_2}\right)
\end{equation}
and the range is $-\infty\leq\varphi\leq \frac{3}{4}$.
For $\lambda_1=-\frac{2}{3}\mu_1$ and $\lambda_2\rightarrow\infty$, we have the limit
\begin{equation}\label{eq:range3}
\varphi=\lim_{\lambda_2\rightarrow\infty}\frac{\left(\mu_1-\mu_2\right)\left(\lambda_2+\tfrac{2}{3}\mu_1\right)}{\tfrac{4}{3}\left(\lambda_2\mu_1-2\mu_2\mu_1\right)}=\frac{3}{4}\left(\frac{\mu_1-\mu_2}{\mu_1}\right)
\end{equation}
with the same range as above.
If we set $\lambda_1$, $\lambda_2\rightarrow\infty$, we obtain the limit equal to zero.
Hence, from above equations~(\ref{eq:range1})--(\ref{eq:range3}), we see that the range of $\varphi$---for all combinations of boundary values of $\lambda_1$, $\lambda_2$, $\mu_1$, $\mu_2$---is
\begin{equation*}
-\infty<\varphi<\frac{3}{4}\,\,.
\end{equation*}
However, using Monte Carlo method, hence, generating random values of $\lambda_1$, $\lambda_2$, $\mu_1$, $\mu_2$, and obeying the stability conditions, we notice that the range of $\varphi$ is $\in\mathbb{R}$.
Therefore, there exist values larger than $\frac{3}{4}$ for some non--boundary Lam\'e parameters.
The analytical proof is the following.
We set
\begin{equation*}
x:=\lambda_1+\mu_2
\qquad
\mathrm{and}
\qquad
y:=\lambda_2+\mu_1
\end{equation*}
and equation~(\ref{eq:phitwolay}) can be rewritten as
\begin{equation*}
\frac{\left(\mu_1-\mu_2\right)\left(\left(y-x\right)-\left(\mu_1-\mu_2\right)\right)}{4\left(xy-\mu_1\mu_2\right)}=\frac{\left(\mu_1-\mu_2\right)\left(\dfrac{\mu_1\mu_2}{x}-x-\left(\mu_1-\mu_2\right)\right)}{4\left(xy-\mu_1\mu_2\right)}\,.
\end{equation*}
In the case of $xy$ being infinitesimally larger than $\mu_1\mu_2$ and very small $x$ (or, similarly, $xy$ being infinitesimally smaller than $\mu_1\mu_2$ and large $x$), this expression tends to $\infty$.
In the case of $xy$ being infinitesimally larger than $\mu_1\mu_2$ and large $x$ (or, similarly, $xy$ being infinitesimally smaller than $\mu_1\mu_2$ and very small $x$), this expression tends to $-\infty$. Hence, the range of $\varphi$ is $\in\mathbb{R}$, as required.
\subsection{Comparison of anisotropy parameters}
Let us compare \citeauthor{Thomsen} parameters and $\varphi$, in a similar way as it was done by~\citeauthor{Berrymanetal}.
\begin{table}[h]
    \centering
    \begin{tabular}{ cccc}
\toprule
 anisotropy & $\lambda=\rm{const}$ & $\lambda+2\mu=\rm{const}$ & $v=\rm{const}$\\
 parameter&$\mu\neq \rm{const}$ &$\mu\neq \rm{const}$ & $\lambda,\mu\neq \rm{const}$ \\
 \cmidrule{1-4}
 $\epsilon$&$\geq0$&$\leq0$ &$\geq0$\\
 $\delta$&$\leq0$&$\leq0$ & $0$\\
 $\gamma$&$\geq0$ & $\geq0$&$\geq0$ \\
 $\varphi$&$0$&$\in \mathbb{R}$ & $\in \mathbb{R}$ \\
 \bottomrule
 \end{tabular}
    \caption{\small{Behaviour of anisotropy parameters as the layer material elasticity parameters vary.}}
    \label{tab:tab1}
\end{table}
In Table~\ref{tab:tab1}, the behaviour of anisotropy parameters is shown.
The situation  of $v=\rm{const}$ corresponds to $\lambda_1\mu_2=\lambda_2\mu_1$, where $v$ is a Poisson's ratio.
The explanation of the behaviour of $\epsilon$ and $\delta$ can be find in the work of \citeauthor{Berrymanetal}.
For instance, results for $\epsilon$ can be shown analytically, while for $\delta$, Monte Carlo method should be used.
The results for $\gamma$ are obtained trivially following Schwarz's inequality and the stability conditions.
The explanation of behaviour of $\varphi$ is shown below.
As it is mentioned in Section~\ref{sec:VarPhi}, for $\lambda={\rm const}$, we have $\varphi=0$.
For $\lambda+2\mu=\rm{const}$, equation~(\ref{eq:phitwolay}) can be rewritten as
\begin{equation*}
\varphi=\frac{2\left(\mu_1-\mu_2\right)^2}{4\left(\lambda_1\lambda_2+\lambda_1\mu_2+\lambda_2\mu_1\right)}=\frac{\left(\mu_1-\mu_2\right)^2}{\left(\lambda_1+2\mu_1\right)\left(\lambda_1+\lambda_2\right)}\,.
\end{equation*}
The numerator is always positive, $\lambda_1+2\mu_1$ is also always positive, but $\lambda_1+\lambda_2$ can be either positive or negative. 
Hence, $\varphi$ is either larger or smaller then zero.
That was also verified numerically, using Monte Carlo method.
In the last case, for constant Poisson's ratio, equation~(\ref{eq:phitwolay}) can be rewritten as
\begin{equation*}
\varphi=\frac{2\lambda_1\mu_2-\lambda_1\mu_1-\lambda_2\mu_2}{4\left(\lambda_1\lambda_2+2\lambda_1\mu_2\right)}
\end{equation*}
and it may be positive or negative, which, similarly to the previous case, may be verified easily using Monte Carlo method.
\subsection{Detecting fluids using anisotropy parameters}\label{sec:fluids}
As it was shown by \citet{Gassmann}, the effects of fluids in porous rocks influence only Lam\'e parameter $\lambda$, not rigidity $\mu$. 
Let us discuss each anisotropy parameter in the context of $\lambda$ and $\mu$ fluctuations in finely layered media.
\citeauthor{Berrymanetal} state that
\begin{quote}
Thomsen's parameter $\epsilon$ is smallest when the variation in the layer $\lambda$ Lam\'e parameter is large, independent of the variation in the shear modulus $\mu$. $[\,\dots\,]$\, 
Similarly, we find that $\delta$ is positive in finely layered media having large variations in $\lambda$.
\end{quote}
The authors conclude that 
\begin{quote}
the regions of small positive $\epsilon$ when occurring together with small positive $\delta$ may be useful indicators of rapid spatial changes in fluid content in the layers.
\end{quote}
\citeauthor{Thomsen} parameter $\gamma$ indicates the variation of $\mu$ in layers.
Its large value indicates significant fluctuation of shear modulus.
Moreover, it does not depend on $\lambda$.
On the other hand, $\varphi$ shows the variation of $\lambda$ in layers, but also depends on the fluctuation of $\mu$.
Combining the information provided by $\gamma$ and $\varphi$ may be useful in detection of changes in fluid content.
In the case of large $\gamma$, very small values of $\varphi$ indicate $\lambda\approx\rm{const}$, and conversely, large values indicate variations of $\lambda$ in layers.

To discuss the anisotropy parameters in the context of the fluctuation of $\lambda$, it might be useful to consider the numerical examples from Tables~\ref{tab:tab2}~and~\ref{tab:tab3}.
For simplification, we examine only five horizontal thin isotropic layers. 
In Table~\ref{tab:tab2}, we have shown four different cases of the fluctuations of $\lambda$ and $\mu$ in layers.
Based on work of~\citet{Ji}, these are possible values for mafic rocks (gabbro, diabase, mafic gneiss, and mafic granulitein) in Earth's crust and upper mantle.
Table~\ref{tab:tab3} reflects anisotropy parameters computed for the \citeauthor{Backus} medium for four consecutive cases from Table~\ref{tab:tab2}.
Both Tables should not be viewed in separation.

\renewcommand{\arraystretch}{1.4}
\begin{table}[h]
\centering
\begin{subtable}[b]{2.7cm}
\centering
\begin{tabular}{cc} 
\multicolumn{2}{c}{\footnotesize{[GPa]}}\\
\toprule 
\hphantom{x}$\lambda$\hphantom{x} & \hphantom{x}$\mu$\hphantom{x} \\ 
\toprule 
85 & 46.8 \\ 
40 & 47.1 \\ 
53 & 46.9 \\ 
80 & 47.0 \\ 
65 & 47.2 \\ 
\bottomrule
\end{tabular}
\caption{}
\label{subtab:a}
\end{subtable}
\,\,\,
\begin{subtable}[b]{2.7cm}
\centering
\begin{tabular}{cc} 
\multicolumn{2}{c}{\footnotesize{[GPa]}}\\ 
\toprule
$\lambda$ & $\mu$ \\ 
\toprule
61.8 & 46.8 \\ 
62.1 & 47.1 \\ 
61.9 & 46.9 \\ 
62.0 & 47.0 \\ 
62.2 & 47.2 \\ 
\bottomrule
\end{tabular}
\caption{}
\label{subtab:b}
\end{subtable}
\,\,\,
\begin{subtable}[b]{2.7cm}
\centering
\begin{tabular}{cc} 
\multicolumn{2}{c}{\footnotesize{[GPa]}}\\ 
\toprule
$\lambda$ & $\mu$ \\ 
\toprule
61.8 & 37 \\ 
62.1 & 56 \\ 
61.9 & 44 \\ 
62.0 & 35 \\ 
62.2 & 52 \\ 
\bottomrule
\end{tabular}
\caption{}
\label{subtab:c}
\end{subtable}
\,\,\,
\begin{subtable}[b]{2.7cm}
\centering
\begin{tabular}{cc} 
\multicolumn{2}{c}{\footnotesize{[GPa]}}\\ 
\toprule
$\lambda$ & $\mu$ \\ 
\toprule
45 & 37 \\ 
70 & 56 \\ 
86 & 44 \\ 
52 & 35 \\ 
64 & 52 \\ 
\bottomrule
\end{tabular}
\caption{}
\label{subtab:d}
\end{subtable}
\caption{\small{
\subref{subtab:a} strongly varying $\lambda$, barely varying $\mu$;
\subref{subtab:b} barely varying both $\lambda$ and $\mu$;
\subref{subtab:c} barely varying $\lambda$, strongly varying $\mu$;
\subref{subtab:d} strongly varying both $\lambda$ and $\mu$
}}
\label{tab:tab2}
\end{table}

\renewcommand{\arraystretch}{1.4}
\begin{table}[h]
\centering
\begin{subtable}[b]{2.7cm}
\centering
\begin{tabular}{cc} 
\multicolumn{2}{c}{\footnotesize{$\times10^{-5}$}}\\
\toprule
$\epsilon$ & -8.11\\ 
$\delta$ & -8.50 \\ 
$\gamma$ & 0.45 \\
$\varphi$ & -10.32 \\ 
\bottomrule
\end{tabular}
\caption{}
\label{subtab:a}
\end{subtable}
\,\,\,
\begin{subtable}[b]{2.7cm}
\centering
\begin{tabular}{cc} 
\multicolumn{2}{c}{\footnotesize{$\times10^{-5}$}}\\
\toprule
$\epsilon$ & 0.33\\ 
$\delta$ & -0.05 \\ 
$\gamma$ & 0.45 \\
$\varphi$ & 0.21 \\ 
\bottomrule
\end{tabular}
\caption{}
\label{subtab:b}
\end{subtable}
\,\,\,
\begin{subtable}[b]{2.7cm}
\centering
\begin{tabular}{cc} 
\multicolumn{2}{c}{\footnotesize{$\times10^{-3}$}}\\
\toprule
$\epsilon$ & 5.91\\ 
$\delta$ & -7.93 \\ 
$\gamma$ & 17.05 \\
$\varphi$ & 0.09 \\ 
\bottomrule
\end{tabular}
\caption{}
\label{subtab:c}
\end{subtable}
\,\,\,
\begin{subtable}[b]{2.7cm}
\centering
\begin{tabular}{cc} 
\multicolumn{2}{c}{\footnotesize{$\times10^{-3}$}}\\
\toprule
$\epsilon$ & 12.38\\ 
$\delta$ & -1.54 \\ 
$\gamma$ & 17.05 \\
$\varphi$ & 7.56 \\ 
\bottomrule
\end{tabular}
\caption{}
\label{subtab:d}
\end{subtable}
\caption{\small{
\subref{subtab:a} Anisotropy parameters for the first;
\subref{subtab:b} second; 
\subref{subtab:c} third and;
\subref{subtab:d} the last case
}}
\label{tab:tab3}
\end{table}
The results from examples~\ref{subtab:c} and~\ref{subtab:d} show clearly that \citeauthor{Thomsen} parameter $\epsilon$ is not smaller for larger fluctuations of $\lambda$. 
Also, from~\ref{subtab:b} and~\ref{subtab:c}, we see that $\epsilon$ significantly depends on $\mu$.
Moreover, for large fluctuations of $\lambda$ from examples~\ref{subtab:a} and~\ref{subtab:d}, $\delta$ is not positive.
It might be positive for different examples, but it is not always the case.
The above considerations contradict the first quote of~\citeauthor{Berrymanetal}.
Since the first quote is not always true, it means that the conclusion in the second quote might be true, but only in certain cases.
In the example~\ref{subtab:d}, the fluctuation of $\lambda$ is high and $\epsilon$ is relatively large, while $\delta$ is negative, which is probably caused by the large variations of $\mu$.
From above examples, it is difficult to identify any pattern of relation between $\epsilon$ and $\delta$ indicating the variations of $\lambda$.

However, examining the examples from Tables~\ref{tab:tab2}~and~\ref{tab:tab3}, we find an alternative way of indicating the variations of $\lambda$ in layers, by omitting the parameter $\delta$.
In the case of barely varying $\mu$, hence, near--constant rigidity, the indicator of change of the fluid content might be the relation between $\epsilon$ and $\varphi$.
If there are no variations of $\lambda$, the absolute value of $\epsilon$ is larger than the one of $\varphi$, as it is exemplified in Table~\ref{subtab:b}.
If there are variations, the absolute value of $\epsilon$ is smaller than $\varphi$, as in Table~\ref{subtab:a}.
In the case of varying share modulus, it is enough to consider $\varphi$ solely.
For no fluctuations of $\lambda$, $\varphi\approx0$, as shown in Table~\ref{subtab:c}.
For varying $\lambda$, as in Table~\ref{subtab:d}, $\varphi$ is relatively large---almost 100 times larger than in the no fluctuation case.
Aforementioned relations are shown in Table~\ref{tab:four}.

\renewcommand{\arraystretch}{1.4}
\begin{table}[h]
\centering
\begin{tabular}{cccc} 
\toprule
\multicolumn{2}{c}{$\gamma\approx0$} & \multicolumn{2}{c}{\hphantom{xxxx}$\gamma>0$}\\
\midrule
$\lambda\approx\rm{const}$ \,\,\,&\,\,\, $\lambda\neq\rm{const}$\,\qquad &\,\qquad $\lambda\approx\rm{const}$ \,\,\,&\,\,\, $\lambda\neq\rm{const}$\\
$|\epsilon|>|\varphi|$ \,\,\,&\,\,\, $|\epsilon|<|\varphi|$\,\qquad &\,\qquad $\varphi\approx0$  \,\,\,&\,\,\, $\varphi\neq0$\\ 
\bottomrule
\end{tabular}
\caption{\small{Pattern for an alternative way of detecting fluids by excluding $\delta$}}
\label{tab:four}
\end{table}

The way of indicating change of fluids in layered Earth shown by~\citeauthor{Berrymanetal} requires further investigation.
The examples shown by us, in which the~\citeauthor{Berrymanetal} method is not applicable, do not indicate its uselessness, and the other examples---in which it might be applicable---should be examined.
Moreover, the case of the large fluctuation of $\lambda$, where $\delta$ is negative and $\epsilon$ positive or negative must be investigated more deeply.
The same regards the pattern for detecting fluid change using $\varphi$.
\section{Conclusions}
In the case of thin isotropic layers, its equivalent medium is isotropic if and only if $\mu=\rm{const}$.
For a PITL medium, the alternative derivation to the one of \citeauthor{Postma} is shown in Section~\ref{sec:pitl}.
For the general case of finely layered medium, we have confirmed that---under the assumption of strict stability conditions---conditions~(\ref{eq:baccond}) from the work of \citeauthor{Backus}, namely,
\begin{equation*}
\overline{c_{2323}}=\overline{\left(\frac{1}{c_{2323}}\right)}^{-1}\,,
\,\,\,
\overline{\left(\frac{\left(c_{2323}\right)^2}{c_{1111}}\right)}=\overline{\left(\frac{c_{2323}}{c_{1111}}\right)}\overline{c_{2323}}\,,
\,\,\,
\overline{\left(\frac{c_{2323}}{c_{1111}}\right)}=\overline{c_{2323}}\overline{\left(\frac{1}{c_{2323}}\right)}\,,
\end{equation*}
are sufficient for the equivalent medium to be isotropic.
An alternative way of checking the isotropy of equivalent medium is to examine if $c^{\overline{\rm TI}}_{1111} = c^{\overline{\rm TI}}_{3333}\,$ and
$c^{\overline{\rm TI}}_{1133} = c^{\overline{\rm TI}}_{1111} - 2\,c^{\overline{\rm TI}}_{1212}$, which means that isotropy conditions are satisfied, since such a medium with $c^{\overline{\rm TI}}_{1212} \neq c^{\overline{\rm TI}}_{2323}\,$ cannot exist.

Parameter, $\varphi$, showing the anisotropy of finely layered medium is introduced.
We have found its range and compared its behaviour---as the elasticity parameters vary---to the one of \citeauthor{Thomsen} parameters.
The dependance of $\epsilon$ on the values of $\gamma$---in contrary to the statement of~\citeauthor{Berrymanetal}---is noticed.
Also, we have exhibited certain cases of large variation of $\lambda$, in which $\epsilon$ has relatively large values and $\delta$ is negative.
Finally, we have shown the examples where the~\citeauthor{Berrymanetal} method of indicating fluids is not applicable.
The possible use of $\varphi$ as an indicator of change of fluid content in layered earth is demonstrated.
The application of $\varphi$ might be treated as an alternative anisotropy parameter to $\delta$, or as a supportive parameter containing additional physical information useful in liquid detection.
\section*{Acknowledgements}
We wish to acknowledge discussions with the supervisor Michael A. Slawinski, the consultations with Len Bos, the computer support of Izabela Kudela and priceless comments along with the editorial work of Theodore Stanoev.
This research was performed in the context of The Geomechanics Project supported by Husky Energy. 
\bibliographystyle{apa}
\bibliography{bibliography}

\begin{thebibliography}{}

\bibitem[\protect\astroncite{Adamus et~al.}{2018}]{Adamusetal}
Adamus, F.~P., Slawinski, M.~A., and Stanoev, T. (2018).
\newblock On effects of inhomogeneity on anisotropy in {B}ackus average.
\newblock {\em arXiv:1802.04075 [physics.geo-ph]}.

\bibitem[\protect\astroncite{Backus}{1962}]{Backus}
Backus, G.~E. (1962).
\newblock Long-wave elastic anisotropy produced by horizontal layering.
\newblock {\em Journal of Geophysical Research}, 67(11).

\bibitem[\protect\astroncite{Berryman et~al.}{1999}]{Berrymanetal}
Berryman, J.~G., Grechka, V.~Y., and Berge, P.~A. (1999).
\newblock Analysis of {T}homsen parameters for finely layered {VTI }media.
\newblock {\em Geophysical Prospecting}, 47(6):959--978.

\bibitem[\protect\astroncite{Bos et~al.}{2017}]{Bosetal}
Bos, L., Dalton, D.~R., Slawinski, M.~A., and Stanoev, T. (2017).
\newblock On {B}ackus average for generally anisotropic layers.
\newblock {\em Journal of Elasticity}, 127:179--196.

\bibitem[\protect\astroncite{Gassmann}{1951}]{Gassmann}
Gassmann, F. (1951).
\newblock {\"U}ber die {E}lastizit{\"a}t por{\"o}ser {M}edien.
\newblock {\em Vierteljahrsschrift der {N}aturforschenden {G}esellschafy in
  {Z}urich}, 96:1--23.

\bibitem[\protect\astroncite{Helbig and Schoenberg}{1987}]{HelbigandSchoenberg}
Helbig, K. and Schoenberg, M. (1987).
\newblock Anomalous polarization of elastic waves in transversely isotropic
  media.
\newblock {\em The Journal of the Acoustical Society of America},
  81(5):1235--1245.

\bibitem[\protect\astroncite{Ji et~al.}{2010}]{Ji}
Ji, S., Sun, S., Wang, Q., and Marcotte, D. (2010).
\newblock Lam{\'e} parameters of common rocks in the {E}arth's crust and upper
  mantle.
\newblock {\em Journal of Geophysical Research}, 115:B06314.

\bibitem[\protect\astroncite{Postma}{1955}]{Postma}
Postma, G.~W. (1955).
\newblock Wave propagation in a stratified medium.
\newblock {\em Geophysics}, 20(4):780--806.

\bibitem[\protect\astroncite{Schoenberg and Muir}{1989}]{SchoenbergandMuir}
Schoenberg, M. and Muir, F. (1989).
\newblock A calculus for finely layered anisotropic media.
\newblock {\em Geophysics}, 54(5):581--589.

\bibitem[\protect\astroncite{Slawinski}{2015}]{SlawinskiRed}
Slawinski, M.~A. (2015).
\newblock {\em Waves and rays in elastic continua}.
\newblock World Scientific, 3 edition.

\bibitem[\protect\astroncite{Slawinski}{2018}]{SlawinskiGreen}
Slawinski, M.~A. (2018).
\newblock {\em Waves and rays in seismology: {A}nswers to unasked questions}.
\newblock World Scientific, 2 edition.

\bibitem[\protect\astroncite{Thomsen}{1986}]{Thomsen}
Thomsen, L. (1986).
\newblock Weak elastic anisotropy.
\newblock {\em Geophysics}, 51(10):1954--1966.

\end{thebibliography}
\end{document}